\documentclass[aps,pra,a4paper,twocolumn,oneside,amsmath,amssymb,footinbib,superscriptaddress]{revtex4-1}

\usepackage{bm}

\usepackage{xcolor}
\usepackage{geometry}

\usepackage{amsmath, amsfonts,amssymb,theorem,
euscript,array,enumerate,amsfonts,mathrsfs,color}

\usepackage[breaklinks=true,colorlinks,citecolor=blue,linkcolor=blue,urlcolor=blue]{hyperref}
\usepackage{dcolumn}
\usepackage{graphics}
\usepackage{epstopdf}
\usepackage[dvips]{epsfig}
\input xy
\xyoption{all}

\definecolor{Ying}{rgb}{0.8,0,0.5}

\newcommand{\bde}{\begin{displaymath}}
\newcommand{\ede}{\end{displaymath}}
\newcommand{\el}{\end{lem}}
\newcommand{\be}{\begin{equation}}
\newcommand{\ee}{\end{equation}}
\newcommand{\beq}{\begin{eqnarray*}}
\newcommand{\eeq}{\end{eqnarray*}}
\newcommand{\beqa}{\begin{eqnarray}}
\newcommand{\eeqa}{\end{eqnarray}}
\newcommand{\bel }{\left\{\begin{array}{ll}}
\newcommand{\eel}{\cr \end{array} \right.}
\newcommand{\bex}{\begin{ex} \rm }
\newcommand{\eex}{\end{ex}}

\newcommand{\bp}{\begin{pro}}
\newcommand{\ep}{\end{pro}}

\def\ee{\epsilon}

\def\lm{\lambda}

\def\beqlb{\begin{eqnarray}}\def\eeqlb{\end{eqnarray}}
\def\beqnn{\begin{eqnarray*}}\def\eeqnn{\end{eqnarray*}}

\theoremstyle{plain}
\theorembodyfont{\upshape\it}
\newtheorem{Thm}{\bf Theorem}[section]
\newtheorem{Pro}[Thm]{\bf Proposition}

\theorembodyfont{\upshape}

\newcommand {\finproof} {\hfill $\Box$ \vskip 5 pt }

\def\edoc{\end{document} }

\begin{document}

\author{Bin Xie}
\email{bin.xie@uga.edu}
\affiliation{University of Georgia, Mathematics Department. Athens, GA 30605, USA}
\affiliation{Beijing Wuzi University, Information College, Tongzhou, Beijing 101149, P.R.China}

\author{Weiping Li}
\email{w.li@okstate.edu}
\affiliation{Southwest Jiaotong University,  Chengdu, Sichuan Province 611756, P.R.China }
\affiliation{Oklahoma State University, Mathematics Department, Stillwater,OK 74078, USA}

\title{The Parameter Sensitivities of a Jump-diffusion Process in Basic Credit Risk Analysis}

\begin{abstract}
We detect the parameter sensitivities of bond pricing which is driven by a Brownian motion and a compound Poisson process as the discontinuous case in credit risk research. The strict mathematical deductions are given theoretically due to the explicit call price formula. Furthermore, we illustrate Matlab simulation to verify these conclusions.

\end{abstract}

\maketitle

\section{Introduction}
\label{intro}
The continuous model of option pricing was rising up in Black and Scholes \cite{BS} on 1973 and the discontinuous one with Compound Poisson Process was studied by Merton \cite{Merton} on 1976. Zhou \cite{Zhou} disclosed the credit risk approach base on this jump-diffusion process on 1997. 

With the explicit Black-Scholes option formulas, the strict derivation of parameters sensitivities were disclosed and named with Greeks. The Greeks are very useful in the stock market due to hedging practices. We also explore the parameter sensitivities for bond price through the complicated jump-diffusion formula according to the log-normal distributed compound Poisson processes strictly in this paper. Beyond the proofs, the analytic illustrations are provided by Matlab simulation.

The rest of the paper is organized as follows. Sections 2 is model framework about the bond price formula of jump-diffusion model with compound Poisson processes. Section 3 is the proportions of parameter sensitivities with strict proofs. Section 4 concludes.

\section{Model Framework}
The financial market has a big part which is called credit market or bond market. The participants can issue new debts or securities on the credit market. Credit risk is the crucial problem for the loss risk of borrower's failure to meet the obligations.

\subsection{Basic Credit Risk Concepts}
Assume that we are in the setting of the standard Black-Scholes model, i.e. we analyze a market with continuous trading which is frictionless and competitive with assumptions \cite{Lando}.\\
1. agents are price takers.\\
2. there are no transaction costs.\\
3. there is unlimited access to short selling and no indivisibilities of assets.\\
4. borrowing and lending through a money-market account can be done at some riskless, continuously compounded rate $r$.

We want to price bonds issued by a firm whose assets are assumed to follow a geometric Brownian motion:
$$dV_t=\mu V_t dt+\sigma V_tdW_t.$$
Here, $W$ is a standard Brownian motion under the probability measure $\textbf{P}$.\\
Let the starting value of assets is $V_0$. Then by Ito-Doeblin formula:
$$V_t=V_0exp((\mu-\frac{1}{2}\sigma^2)t+\sigma W_t).$$

We take it to be well known that in an economy consisting of these two assets, the  price $C_0$ at time $0$ of a contingent claim paying $C(V_T)=C_T$ at time $T$ is equal to 
$$C_0=E^Q[e^{-rt}C_T]  ,$$
where \textbf{Q} is the equivalent martingale measure under which the dynamics of $V$ are given as
$$V_t=V_0exp((r-\frac{1}{2}\sigma^2)t+\sigma W^Q_t).$$
Here, $W^Q_t$ is a Brownian motion and we can see that the drift term $\mu$ has been replaced by $r$.\cite{Lando}

Now, assume that the firm at time $0$ has issued two types of claims: debt and equity. In the simple model, debt is a zero-coupon bond with a face value of $D$ and maturity date $T$. We think of the firm run by the equity owners. At maturity of bond, equity holder pay the face value of debt precisely when the assets value is higher than the face value of the bond.
On the other hand, if assets are worth less than $D$, equity owners do not want to pay $D$. And since they have limited liability they don't have to do that. Bond holders then take over the remaining assets of $V_T$ instead of the promised payment $D$. 
With this assumption, the payoffs to debt, $B_T$, and equity, $S_T$, at date $T$ are given as:
\begin{gather*}
B_T=min(D,V_T)=D-max(D-V_T,0),\\
S_T=max(V_T-D,0).
\end{gather*}
From the structure, debt can be viewed as the difference between a riskless bond and a put option, and equity can be viewed as a call option on the firm's assets. \cite{Lando}\\

We assumed there are no transaction costs, bankruptcy costs, taxes and so on for simpleness. We then get $V_T=B_T+S_T$.
Given the current level $V$ and volatility $\sigma$ of assets, and the riskless rate $r$, we denote the Black-Scholes model of European call as $C(V_t,D, \sigma,r,T-t)$ with strike price $D$ and maturity time $T$, \cite{Lando} i.e.
$$C(V_t,D, \sigma,r,T-t)=V_tN(d_1)-De^{-r(T-t)}N(d_2).$$
Where $N$ is the standard normal distribution function and 
\begin{gather*}
    d_{1,2}=\dfrac{ln(V_t/D)+(r\pm\dfrac{1}{2}\sigma^2)(T-t)}{\sigma\sqrt{T-t}},\\
    d_1-d_2=\sigma\sqrt{T-t}.
\end{gather*}

Applying the Black-Scholes formula to price these options, we obtain the Merton model for values of debt and equity at time t as:
\begin{gather*}
S_t=C(V_t,D, \sigma,r,T-t), \\
B_t=De^{(-r(T-t))}-P(V_t,D,\sigma,r,T-t)
\end{gather*}

From the put-call parity for European options on non-dividend paying stocks
$$C(V_t,D, \sigma,r,T-t)-P(V_t,D,\sigma,r,T-t)=V_t-De^{-r(T-t)}.$$

We get
\begin{align*}
B_t &= De^{(-r(T-t))}-P(V_t,D,\sigma,r,T-t)\\
&=V_t-C(V_t,D, \sigma,r,T-t)\\
&=V_t(1-N(d_1))+De^{-r(T-t)}N(d_2).
\end{align*}

\subsection{Basic Credit Risk Analysis with Compound Poisson Jumps}
For the discontinuous Black-Scholes model, we may consider the case of compound Poisson jumps which has the explicit formula for the the call price. Therefore, the bond price is obvious from the call-put parity and equality $V_T=B_T+S_T$ at maturity time $T$.

Suppose asset value $V_t$ has dynamics of jumps, then the equity value $S_t$ is priced as a call option $C^J$ with jumps.
First, we focus the compound Poisson jumps with i.i.d. log-normal distributed $Y_i+1$ (i.e., $ln(Y_i+1)\sim N(\mu,\delta ^2)$) which has the explicit formula, the price of call option $C^J$ is as following \cite{Merton}:
\begin{multline} \label{call}
C^J(V_t,D,\tau,\sigma^2,r,\delta^2,\lm,k)\\
  =\sum ^{\infty}_{n=0} \frac{(\lm'\tau)^n}{n!}
e^{-\lm' \tau} C(V_t,D,\tau,\sigma_n,r_n)
\end{multline} 

where $\lm$ is the intensity of Poisson process, $C(V_t,D,\tau,\sigma_n,r_n)$ is the standard Black-Scholes formula for a call and 
\begin{gather*}
k =E(Y_i),\\
\lm'=\lm (1+k),\\
r_n=r+n\gamma/\tau-\lm k,\\
\sigma_n^2 = \sigma^2+n\delta^2/\tau,\\
\gamma = ln(1+k)= \mu+\frac{1}{2}\delta^2.
\end{gather*}

In advance, some facts of general Black-Scholes call price are listed \cite{JH}:
\begin{gather*}
    C_x =N(d_1)=\Delta>0,\\
C_{\tau}  = \frac{S_t\sigma}{2\sqrt{\tau}}n(d_1)+Kre^{-r\tau}N(d_2)=\Theta>0,\\
C_{\sigma}  = S_t\sqrt{\tau}n(d_1)=Vega>0,\\
C_r  = \tau  e^{-r\tau}N(d_2)=Rho>0,\\
C_K  = - e^{-r\tau}N(d_2)<0.
\end{gather*}

\section{Sensitivities of Bond Pricing for Log-normal Jumps Process}

Due to the explicit formula, the derivatives with respect to all parameters are examined as the sensitivities of bond price.

\begin{Pro}
(i) The bond price is increasing in $V_t$ for log-normal jumps process.

(ii) $ \dfrac{\partial B_t}{\partial x}\in (0,1).$
\end{Pro}

\begin{proof}

Let $V_t=x$, $T-t= \tau.$
Then we check the partial derivative of $B_t$ with respect to $x$.
Hence,
\begin{align*}
&\dfrac{\partial B_t}{\partial x} =\dfrac{\partial}{\partial x}(V_t-C^J(V_t,D,\tau,\sigma^2,r,\delta^2,\lm,k))\\
    &=1-\dfrac{\partial}{\partial x}(\sum ^{\infty}_{n=0} \frac{(\lm'\tau)^n}{n!} e^{-\lm' \tau} C(V_t,D,\tau,\sigma_n,r_n))\\    
    &=1-\sum ^{\infty}_{n=0} \frac{(\lm'\tau)^n}{n!}
e^{-\lm' \tau} \dfrac{\partial C(V_t,D,\tau,\sigma_n,r_n)}{\partial x}) \\
	&=1-\sum ^{\infty}_{n=0} \frac{(\lm'\tau)^n}{n!} e^{-\lm' \tau} N(d_{1n})\\
	&=\sum ^{\infty}_{n=0} \frac{(\lm'\tau)^n}{n!} e^{-\lm' \tau} (1-N(d_{1n}))
	\in (0,1).
\end{align*}
Where $\sum ^{\infty}_{n=0} \frac{(\lm'\tau)^n}{n!} e^{-\lm' \tau}=1$ is convergent.
\end{proof}
\finproof

It is clear that the bond price goes up as $V_t$ increases. We can check this trend by Matlab numerically in using formula \ref{call} (See Figure \ref{fig:Bt-Vt}).

\begin{figure}[htbp]
\centering
\includegraphics[width = .5\textwidth]{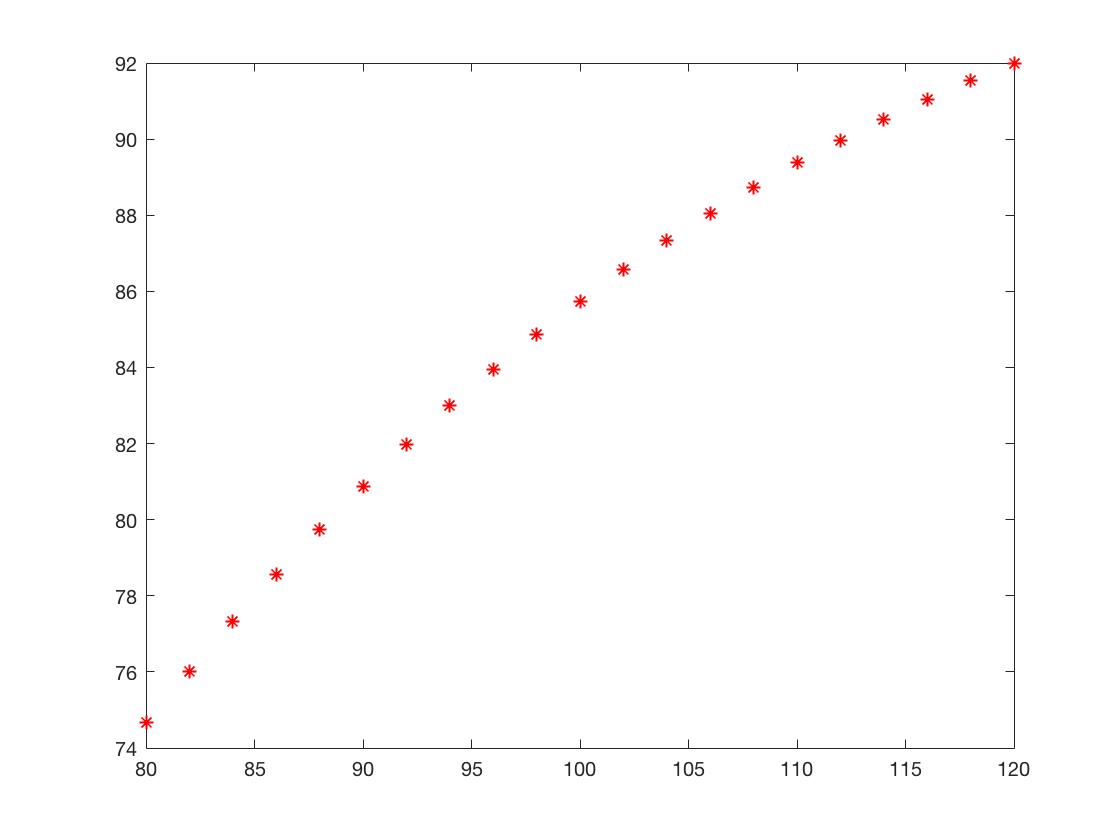}
\caption{Bond price - $V_t$. ($V_t$ is from $80$ to $120$ with step size $2$, $D=110,$ $\tau=2,$ $\sigma=0.2,$ $r=0.05,$ $\lambda=0.1,$ $\mu=-0.2,$ $\delta=0.6,$ upbound of summation $n=50.$) }
\label{fig:Bt-Vt}
\end{figure}

\begin{Pro} 
(i) The bond price is increasing in face value $D$ for log-normal jumps process.

(ii)$\dfrac{\partial B_t}{\partial D} \in (0,e^{-(r -\lm k )\tau})$.
\end{Pro}

\begin{proof}
  
In fact,
\begin{align*}
&\dfrac{\partial C_J(V_t,D,\tau,\sigma^2,r,\delta^2,\lm,k)}{\partial D}\\
    &=\dfrac{\partial}{\partial D}(\sum ^{\infty}_{n=0} \frac{(\lm'\tau)^n}{n!} e^{-\lm' \tau} C(V_t,D,\tau,\sigma_n,r_n))  \\
    &=\sum ^{\infty}_{n=0} \frac{(\lm'\tau)^n}{n!}
e^{-\lm' \tau} \dfrac{\partial C(V_t,D,\tau,\sigma_n,r_n)}{\partial D}\\ 
    &=-\sum ^{\infty}_{n=0} \frac{(\lm'\tau)^n}{n!}
e^{-\lm' \tau} e^{-r_n\tau}N(d_{2n})\\
	&=-\sum ^{\infty}_{n=0} \frac{(\lm'\tau)^n}{n!}
e^{-\lm' \tau} e^{-(r\tau +n\gamma-\lm k \tau)}N(d_{2n})\\
	&= -e^{-(r -\lm k )\tau} \sum ^{\infty}_{n=0} \frac{(\lm'\tau)^n}{n!}
	e^{-\lm' \tau}e^{-n \gamma}N(d_{2n})\\
	&\in (-e^{-(r -\lm k )\tau}, 0).
\end{align*}
Then,$\dfrac{\partial B_t}{\partial D}=-\dfrac{\partial C_J(V_t,D,\tau,\sigma^2,r,\delta^2,\lm,k)}{\partial D}\in (0,e^{-(r -\lm k )\tau})$.

\end{proof} 
\finproof

Definitely, increasing the face value typically will produce a larger payoff.
(See Figure \ref{fig:Bt-D}).

\begin{figure}[htbp]
\centering
\includegraphics[width = .5\textwidth]{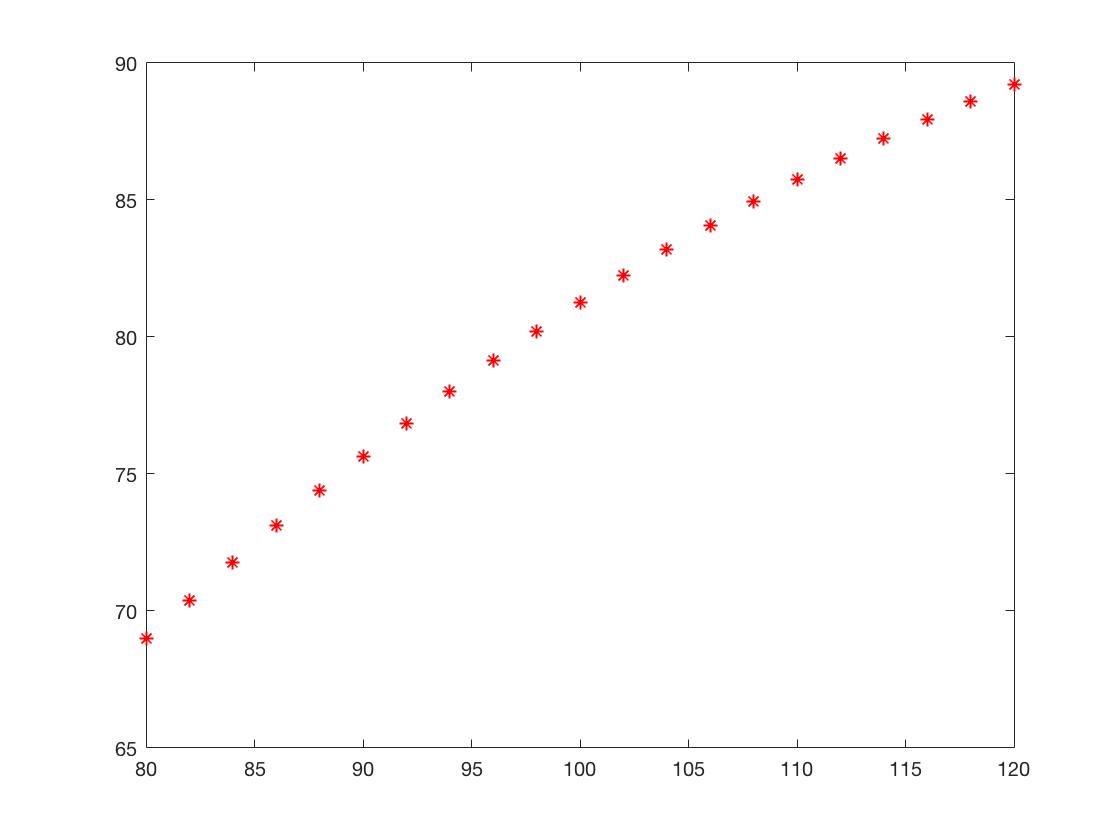}
\caption{Bond price - D. (D is from $80$ to $120$ with step size $2$, $V_t=100,$ $\tau=2,$ $\sigma=0.2,$ $r=0.05,$ $\lambda=0.1,$ $\mu=-0.2,$ $\delta=0.6,$ upbound of summation $n=50.$) }
\label{fig:Bt-D}
\end{figure}

\begin{Pro} 
The bond price is decreasing in volatility, $\sigma$, for log-normal jumps process.
\end{Pro}

\begin{proof}

Mathematically,
\begin{align*}
&\dfrac{\partial C_J(V_t,D,\tau,\sigma^2,r,\delta^2,\lm,k)}{\partial \sigma}\\
    &=\dfrac{\partial}{\partial \sigma}(\sum ^{\infty}_{n=0} \frac{(\lm'\tau)^n}{n!} e^{-\lm' \tau} C(V_t,D,\tau,\sigma_n,r_n))  \\
    &=\sum ^{\infty}_{n=0} \frac{(\lm'\tau)^n}{n!}
e^{-\lm' \tau} \dfrac{\partial C(V_t,D,\tau,\sigma_n,r_n)}{\partial \sigma_n} \\
    &=\sum ^{\infty}_{n=0} \frac{(\lm'\tau)^n}{n!}
e^{-\lm' \tau} x\sqrt{\tau} n(d_{1n})\\
	& =x\sqrt{\tau}\sum ^{\infty}_{n=0} \frac{(\lm'\tau)^n}{n!}
e^{-\lm' \tau}  n(d_{1n})>0.
\end{align*}
Then, $\dfrac{\partial B_t}{\partial \sigma}=-\dfrac{\partial C_J(V_t,D,\tau,\sigma^2,r,\delta^2,\lm,k)}{\partial \sigma}<0$. 

\end{proof} 
\finproof

When the volatility goes up, $B_t$ must decrease because the sum of $S_t$ and $B_t$ remains unchanged, call price increases as $V_t$ is more fluctuable.
(See Figure \ref{fig:Bt-sigma}).

\begin{figure}[htbp]
\centering
\includegraphics[width = .5\textwidth]{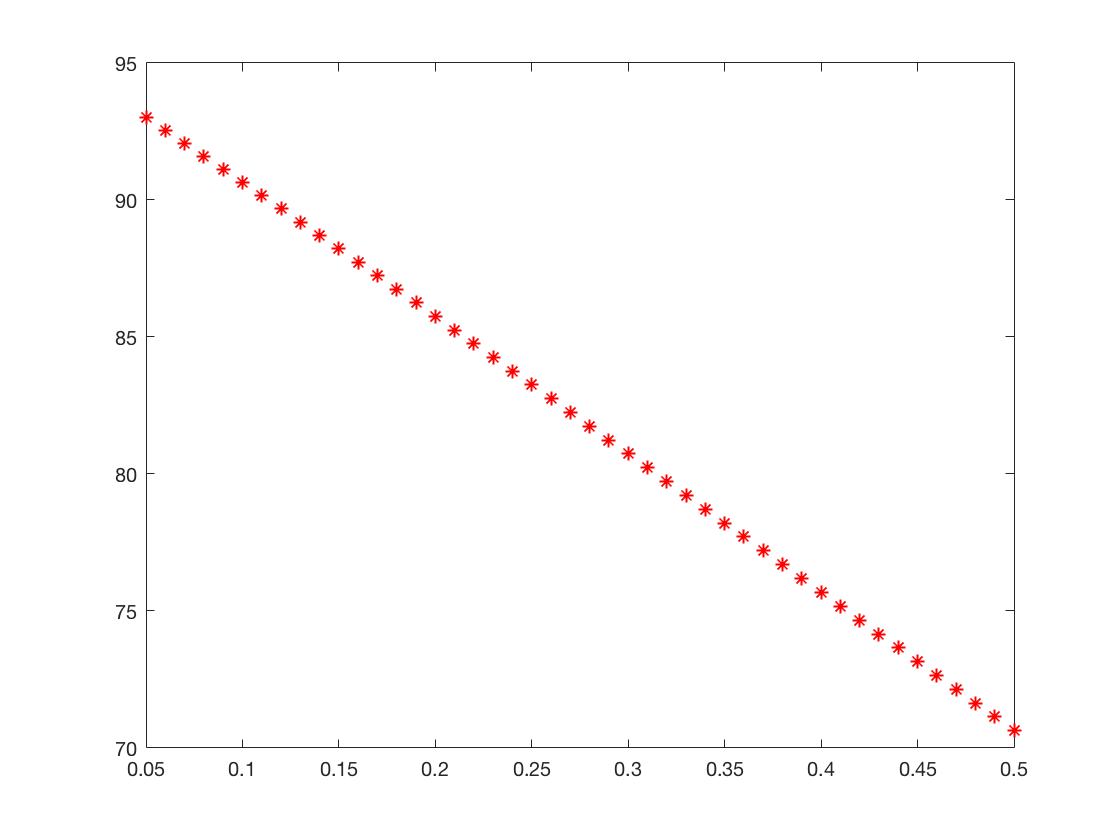}
\caption{Bond price - $\sigma$. ($\sigma$ is from $0.05$ to $0.5$ with step size $0.01$, $V_t=100,$ $D=110,$ $\tau=2,$ $r=0.05,$ $\lm= 0.1,$ $\mu=-0.2,$ $\delta=0.6,$ upbound of summation $n=50.$)}
\label{fig:Bt-sigma}
\end{figure}

\begin{Pro} 
(i) The bond price is decreasing in risk-free interest rate $r$ for log-normal jumps process.

(ii) $\dfrac{\partial B_t}{\partial r} \in (-\tau De^{-(r -\lm k )\tau},0)$.
\end{Pro}

\begin{proof}

Actually,
\begin{align*}
&\dfrac{\partial C_J(V_t,D,\tau,\sigma^2,r,\delta^2,\lm,k)}{\partial r}\\
    &=\dfrac{\partial}{\partial r}(\sum ^{\infty}_{n=0} \frac{(\lm'\tau)^n}{n!} e^{-\lm' \tau} C(V_t,D,\tau,\sigma_n,r_n))  \\
    &=\sum ^{\infty}_{n=0} \frac{(\lm'\tau)^n}{n!}
e^{-\lm' \tau} \dfrac{\partial C(V_t,D,\tau,\sigma_n,r_n)}{\partial r_n} \\
    &=\sum ^{\infty}_{n=0} \frac{(\lm'\tau)^n}{n!}
e^{-\lm' \tau}\tau D e^{-r_n\tau}N(d_{2n})\\
	& = \tau De^{-(r -\lm k )\tau} \sum ^{\infty}_{n=0} \frac{(\lm'\tau)^n}{n!}
	 e^{-n\gamma}N(d_{2n})\\
	 &\in (0,\tau De^{-(r -\lm k )\tau}).
\end{align*}
Then, $\dfrac{\partial B_t}{\partial r}=-\dfrac{\partial C_J(V_t,D,\tau,\sigma^2,r,\delta^2,\lm,k)}{\partial r}
\in (-\tau De^{-(r -\lm k )\tau},0)$.

\end{proof} 
\finproof

Since the call option increases as $r$ goes up, $B_t$ must decrease the money market looks more attractive.(See Figure \ref{fig:Bt-r}).

\begin{figure}[htbp]
\centering
\includegraphics[width = .5\textwidth]{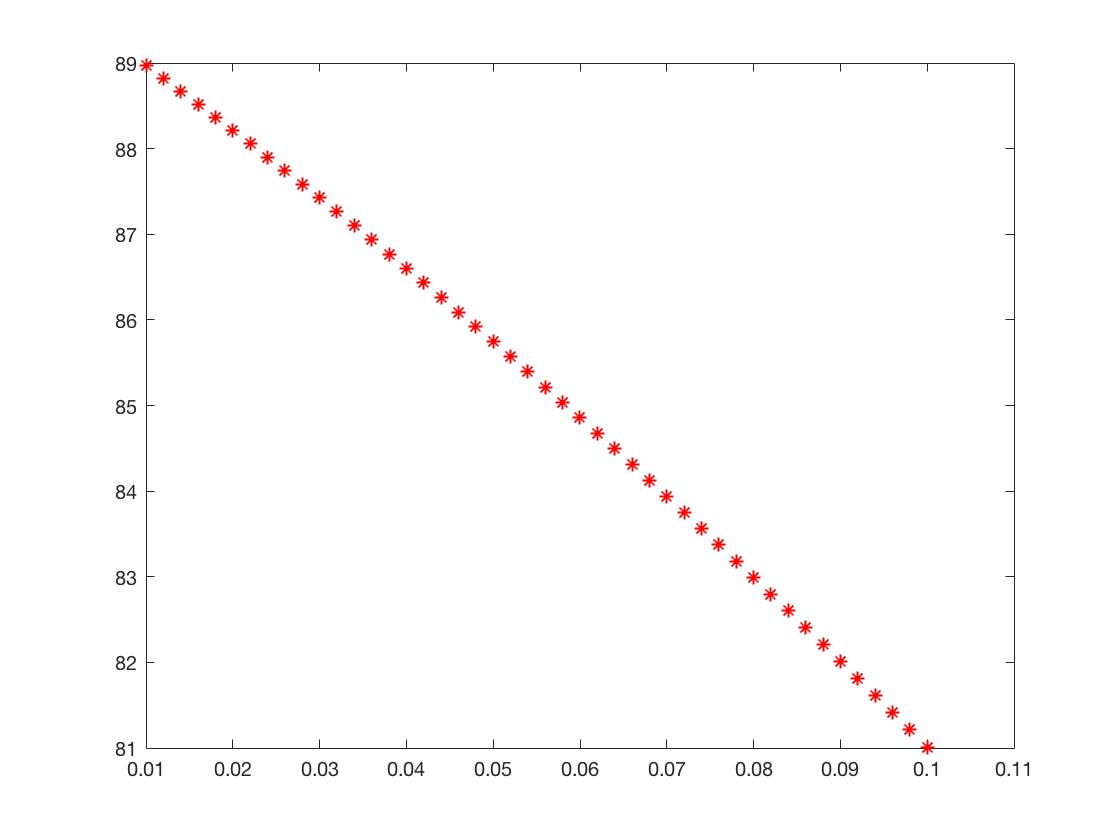}
\caption{Bond price - $r$. ( $r$ is from $0.01$ to $0.1$ with step size $0.002$, $V_t=100,$ $D=110,$ $\tau=2,$ $\sigma=0.2,$ $\lm= 0.1,$ $\mu=-0.2,$ $\delta=0.6,$ upbound of summation $n=50.$)}
\label{fig:Bt-r}
\end{figure}

\begin{Pro}\label{app1} 
(i) The call price is increasing in time-to-maturity, $\tau$, for the log-normal jumps process if $r-\lm k \geq 0$.

(ii) The bond price is decreasing in time-to-maturity, $\tau$, for the log-normal jumps process $r-\lm k \geq 0$.

\end{Pro}

\begin{proof} 
See details at Appendix part 1.

We have $\frac{\partial B_t}{\partial k}<0$.
$B_t$ is decreasing due to the value of call increases when time-to-maturity is bigger.

In some extreme case, when $S_2$ is very small as a negative number, the sum of $S_1$ and $S_2$ can be negative which causes the tendency of $B_t$ w.r.t. $\tau$ is not decreasing.

\end{proof} 
\finproof

These two illustrations are listed as figure \ref{fig:Bt-tau} and figure \ref{fig:Bt-taubad}.

\begin{figure}[!htbp]
\centering
\includegraphics[width = .5\textwidth]{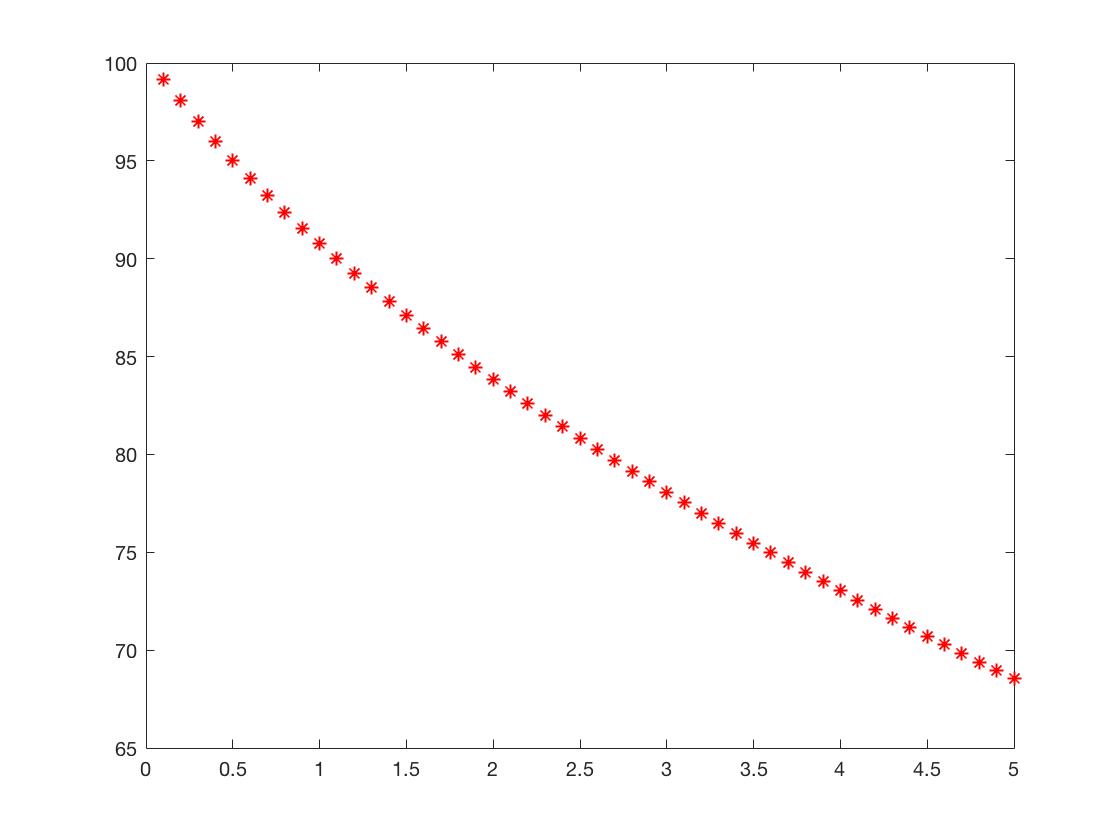}
\caption{Bond price - $\tau$ with condition satisfied. ($\tau$ is from $0.1$ to $5$ with step size $0.1$, $V_t=100,$ $D=110,$ $\sigma=0.2,$ $r=0.05,$ $\lm= 0.1,$ $\mu=-0.2,$ $\delta=0.6,$ upbound of summation $n=50.$) }
\label{fig:Bt-tau}
\end{figure}

\begin{figure}[!htbp]
\centering
\includegraphics[width = .5\textwidth]{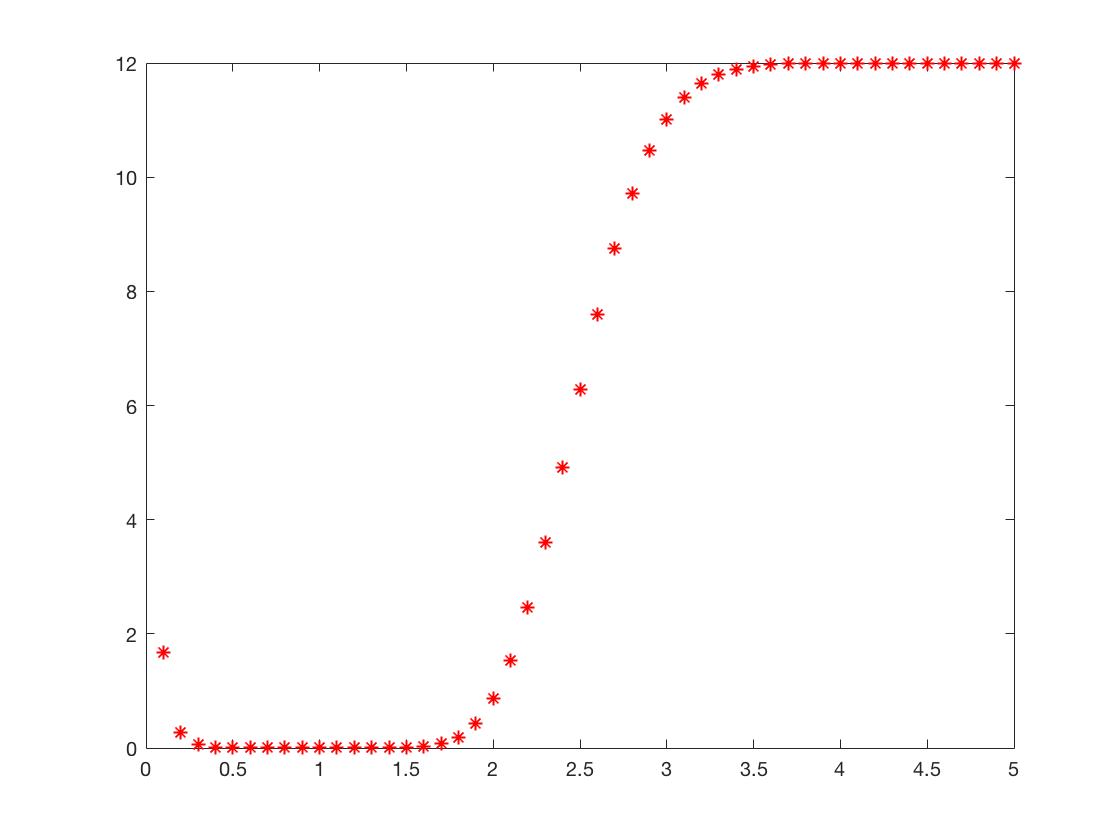}
\caption{Bond price - $\tau$ trend is somehow increasing when condition $r-\lm k\geq 0$  is not satisfied at extreme case. (Here $k=199.34$ under $\mu=0.8, \delta = 3$. $\lm = 0.1,r=0.05$. $r-\lm k =-19.88$, $V_t=12,D=10$. $\tau$ is from 0.1 to 5 with step pace 0.1.)}
\label{fig:Bt-taubad}
\end{figure}

\begin{Pro} 
The bond price is decreasing in $\delta$ for the log-normal jumps process.
\end{Pro}

\begin{proof}

We have
\begin{align*}
&\dfrac{\partial C_J(V_t,D,\tau,\sigma^2,r,\delta^2,\lm,k)}{\partial \delta}\\
    &=\dfrac{\partial}{\partial \delta}(\sum ^{\infty}_{n=0} \frac{(\lm'\tau)^n}{n!} e^{-\lm' \tau} C(V_t,D,\tau,\sigma_n,r_n))  \\
    &=\sum ^{\infty}_{n=0} \frac{(\lm'\tau)^n}{n!}
	e^{-\lm' \tau} \dfrac{\partial C(V_t,D,\tau,\sigma_n,r_n))}{\partial \sigma_n} 
	\dfrac{\partial \sigma_n} {\partial \delta} \\
    &=\sum ^{\infty}_{n=0} \frac{(\lm'\tau)^n}{n!}
	e^{-\lm' \tau}x\sqrt{\tau}n(d_{1n})\frac{2n\delta}{\tau}\\
	& =\frac{2x\delta}{\sqrt{\tau}} \sum ^{\infty}_{n=0} \frac{(\lm'\tau)^n}{n!}
	e^{-\lm' \tau}n(d_{1n})n  >0.
\end{align*}
Then,$\dfrac{\partial B_t}{\partial \delta}=-\dfrac{\partial C_J(V_t,D,\tau,\sigma^2,r,\delta^2,\lm,k)}{\partial \delta}<0$.

\end{proof}
\finproof

See Figure \ref{fig:Bt-delta} for the simulation.

\begin{figure}[!htbp]
\centering
\includegraphics[width = .5\textwidth]{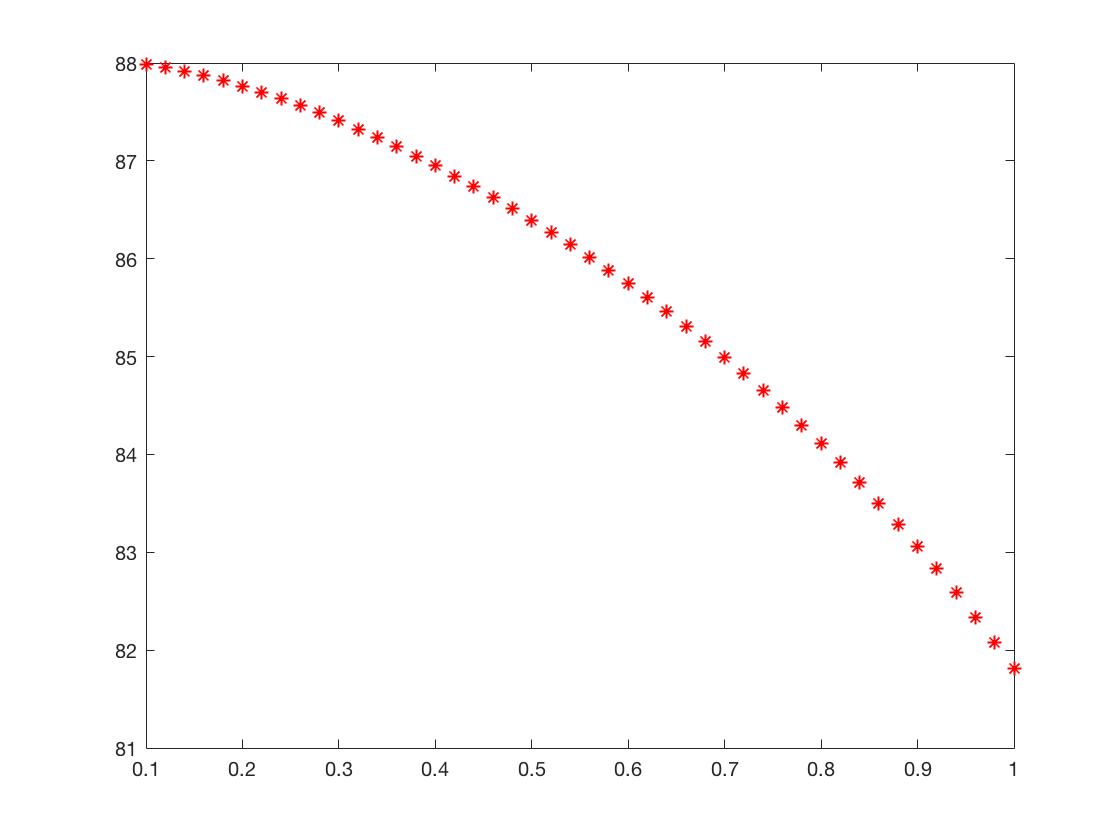}
\caption{Bond price - $\delta$. ($\delta$ is from $0.01$ to $1$ with step size $0.02$, $V_t=100,$ $D=110,$ $\tau=2,$ $\sigma=0.2,$ $r=0.05,$ $\lm= 0.1,$ $\mu=-0.2,$ upbound of summation $n=50.$)}
\label{fig:Bt-delta}
\end{figure}

\begin{Pro}
The bond price is decreasing in $\lambda$ for the log-normal jumps process.
\end{Pro}

\begin{proof}
See details at Appendix part 2.
Finally, 

 $\frac{\partial B_t}{\partial \lm}<0$ holds.

\end{proof}
\finproof

In fact, since $\lambda$ is the intensity of Poisson process, $\lambda$ increasing means the jump rate is bigger, the fluctuation will make the call price go up, so the bond price is decreasing accordingly. (See Figure \ref{fig:Bt-lambda})

\begin{figure}[htbp]
\centering
\includegraphics[width = .5\textwidth]{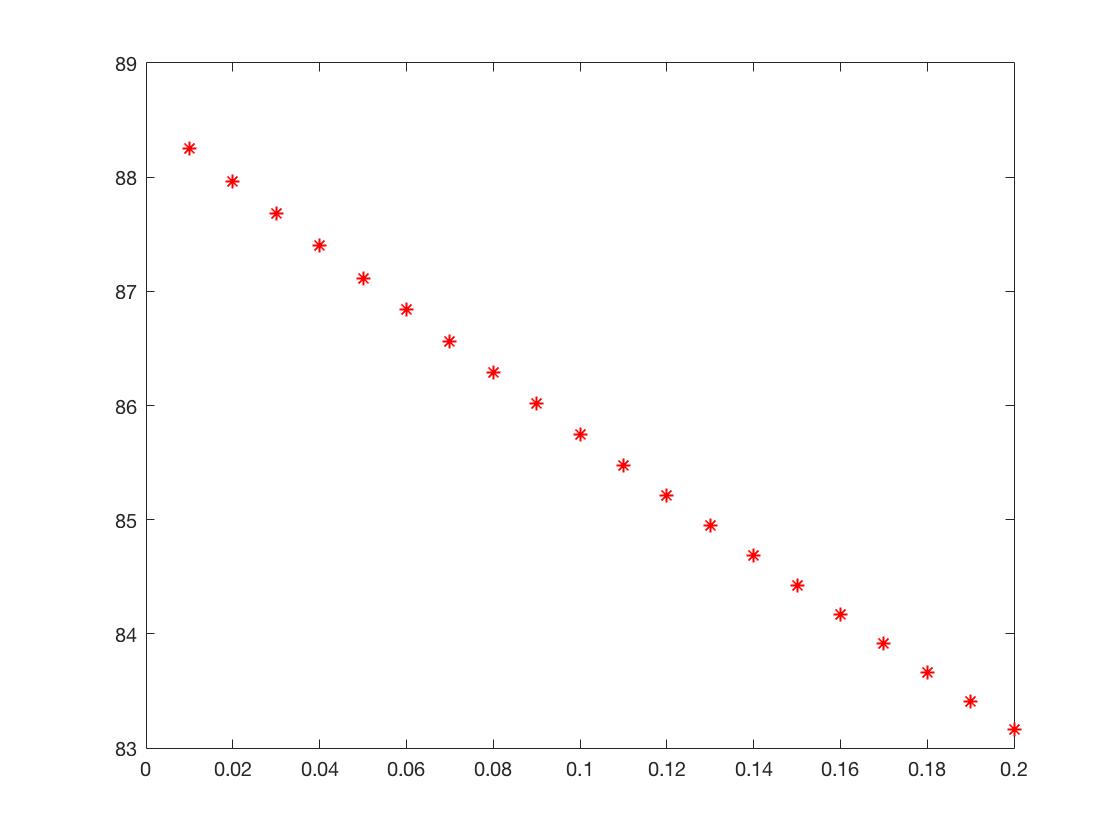}
\caption{Bond price - $\lm$. ($\lm$ is from $0.01$ to $0.2$ with step size $0.01$, $V_t=100,$ $D=110,$ $\tau=2,$ $\sigma=0.2,$ $r=0.05,$ $\mu=-0.2,$ $\delta=0.6,$ upbound of summation $n=50.$)}
\label{fig:Bt-lambda}
\end{figure}

\begin{Pro} 
The bond price is decreasing in $k$ for the log-normal jumps process.
\end{Pro}

\begin{proof}
See details at Appendix part 3.

The result is $\frac{\partial B_t}{\partial k}<0$.

\end{proof}
\finproof

\begin{figure}[!htbp]
\centering
\includegraphics[width = .5\textwidth]{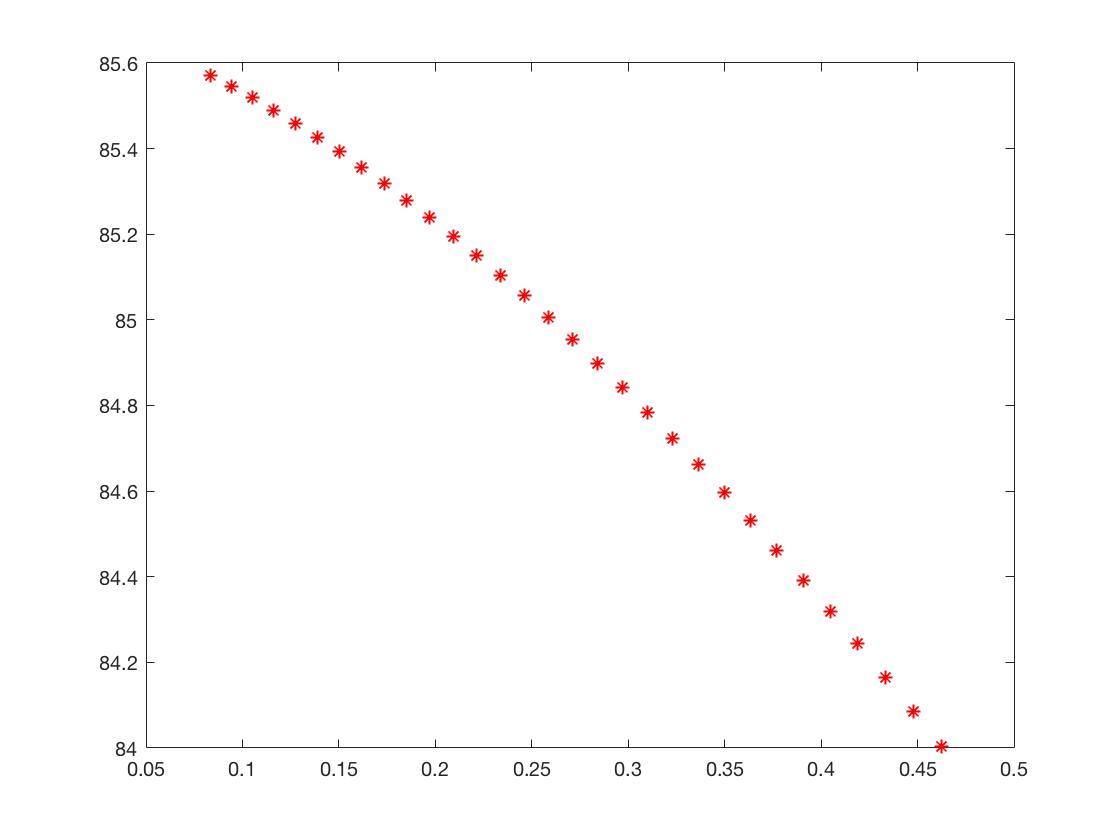}
\caption{Bond price - k. ($k=exp(\mu+\frac{1}{2}\sigma^2)-1,$ where $\mu$ is from $-0.1$ to $0.2$ with step size $0.01$, $V_t=100,$ $D=110,$ $\tau=2,$ $\sigma=0.2,$ $r=0.05,$ $\lm= 0.1,$ $\delta=0.6,$ upbound of summation $n=50.$)}
\label{fig:Bt-k}
\end{figure}

\section{Conclusion}

In this paper, we analytically explored the parameter sensitivities of bond price driven by compound Poisson process with log-normal distributed jumps. Meanwhile, we disclosed the corresponding interval of the sensitivities if available. Generally, the bond price in this case is increasing with respect to asset value $V_t$ and face value $D$, decreasing with respect to volatility $\sigma$, risk-free interest rate $r$,  log-normal standard deviation $\delta$, Poisson intensity $\lambda$ and jumps' mean $k$. For the time-to-maturity $\tau$, it is decreasing when $r-\lm k \geq 0$ and undetermined for other cases.

\newpage
\section*{Appendix}

\textbf{1. Proof of Proposition III.5}

(i) The call price is increasing in time-to-maturity, $\tau$, for the log-normal jumps process if $r-\lm k \geq 0$.

(ii) The bond price is decreasing in time-to-maturity, $\tau$, for the log-normal jumps process $r-\lm k \geq 0$.

\begin{proof}
\begin{align*}
&\dfrac{\partial C^J(V_t,D,\tau,\sigma^2,r,\delta^2,\lm,k)}{\partial \tau}\\
    &=\dfrac{\partial}{\partial \tau}(\sum ^{\infty}_{n=0} \frac{(\lm'\tau)^n}{n!} e^{-\lm' \tau} C(V_t,D,\tau,\sigma_n,r_n))  \\
    &=\sum^{\infty}_{n=1} [n(\lm' \tau)^{n-1} \frac{\lm'}{n!}e^{-\lm'\tau}C_n(V_t,D,\tau,\sigma_n,r_n)]\\ 
    & \quad +\sum^{\infty}_{n=0} [\frac{(\lm' \tau)^n}{n!}(-\lm')e^{-\lm' \tau}C_n(V_t,D,\tau,\sigma_n,r_n)]\\
    & \quad+\sum^{\infty}_{n=0} [\frac{(\lm' \tau)^n}{n!}e^{-\lm' \tau}   ( \frac{\partial C(V_t,D,\tau,\sigma_n,r_n)}{\partial \tau}\\
    & \quad+ \frac{\partial C(V_t,D,\tau,\sigma_n,r_n)}{\partial r_n}\frac{\partial r_n}{\partial \tau}\\
    & \quad+ \frac{\partial C(V_t,D,\tau,\sigma_n,r_n)}{\partial \tau}\frac{\partial \sigma_n}{\partial \tau})]\\
	& =  \sum^{\infty}_{n=1} [ \frac{(\lm' \tau)^{n-1}}{(n-1)!}\lm'e^{-\lm'\tau}C_n(V_t,D,\tau,\sigma_n,r_n)]  \\
	& \quad+\sum^{\infty}_{n=0} [\frac{(\lm' \tau)^n}{n!}(-\lm')e^{-\lm' \tau}C_n(V_t,D,\tau,\sigma_n,r_n)]\\
	& \quad+\sum^{\infty}_{n=0} [\frac{(\lm' \tau)^n}{n!}e^{-\lm' \tau}   ( \frac{\partial C(V_t,D,\tau,\sigma_n,r_n)}{\partial \tau}\\
	& \quad+ \frac{\partial C(V_t,D,\tau,\sigma_n,r_n)}{\partial r_n}\frac{\partial r_n}{\partial \tau}\\
	& \quad+ \frac{\partial C(V_t,D,\tau,\sigma_n,r_n)}{\partial \tau}\frac{\partial \sigma_n}{\partial \tau})]\\
	&=  \sum^{\infty}_{n=0} [ \frac{(\lm' \tau)^{n}}{n!}\lm'e^{-\lm'\tau}(C_{n+1}(V_t,D,\tau,\sigma_n,r_n)\\
	& \quad -C_n(V_t,D,\tau,\sigma_n,r_n))]\\
	& \quad +\sum^{\infty}_{n=0} [\frac{(\lm' \tau)^n}{n!}e^{-\lm' \tau}   ( \frac{\partial C(V_t,D,\tau,\sigma_n,r_n)}{\partial \tau} \\
	& \quad + \frac{\partial C(V_t,D,\tau,\sigma_n,r_n)}{\partial r_n}\frac{\partial r_n}{\partial \tau}\\
	& \quad + \frac{\partial C(V_t,D,\tau,\sigma_n,r_n)}{\partial \tau}\frac{\partial \sigma_n}{\partial \tau})]\\
	&= S_1+S_2,
\end{align*}

Separately, we consider the 1st part $S_1$ and 2nd part $S_2$. 

For $S_1$, since
\begin{align*}
&\dfrac{\partial C(V_t,D,\tau,\sigma_n,r_n)}{\partial n}\\
&= xN'(d_{1n})\dfrac{\partial d_{1n}}{\partial n}
	-De^{-r_n\tau}(-\dfrac{\partial r_n}{\partial n})\tau N(d_{2n}) \\
	& \quad+ De^{-r_n\tau}N'(d_{2n}) \dfrac{\partial d_{2n}}{\partial n}\\
	&=xN'(d_{1n})(\dfrac{\partial d_{1n}}{\partial n}
	-\dfrac{\partial d_{2n}}{\partial n})+De^{-r_n\tau}\gamma \tau N(d_{2n})\\
	&>0
\end{align*}
Where we used the facts:
\begin{gather*}
    xN'(d_{1n})=Ke^{-r_n \tau}N'(d_{2n}),\\
    d_{1n}-d_{2n}=\sigma_n\sqrt{\tau}\Rightarrow
    \dfrac{\partial d_{1n}}{\partial n}-\dfrac{\partial d_{2n}}{\partial n}=\frac{1}{2\sqrt{\sigma_n}} \dfrac{\delta^2}{\tau}\sqrt{\tau}>0.
\end{gather*}

  That means the function $C(V_t,D,\tau,\sigma_n,r_n)$ is increasing w.r.t. $n$, such that $S_1$ part is positive.
  
  For $S_2$, we taking partial derivative of $C(V_t,D,\tau,\sigma_n,r_n)$ w.r.t. $\tau$,
  \begin{align*}
  &S_2= \frac{\partial C(V_t,D,\tau,\sigma_n,r_n)}{\partial \tau}+ \frac{\partial C(V_t,D,\tau,\sigma_n,r_n)}{\partial r_n}\frac{\partial r_n}{\partial \tau}\\
   & \quad + \frac{\partial C(V_t,D,\tau,\sigma_n,r_n)}{\partial \tau}\frac{\partial \sigma_n}{\partial \tau}\\
  & = \frac{x\sigma_n}{2\sqrt{\tau}}n(d_{1n})+Dr_ne^{-r_n\tau}N(d_{2n})\\
  & \quad +\tau D e^{-r_n\tau}N(d_{2n})[n\gamma(-\tau^{-2})]\\
  & \quad +x\sqrt{\tau}n(d_{1n})[\frac{1}{2}(\sigma_n^2)^{-\frac{1}{2}}n\delta^2(-\tau^{-2})]\\
  & = \frac{1}{2\sqrt{\tau}}xn(d_{1n})\frac{\sigma_n^2\tau-n\delta^2}{\sigma_n\tau}
  +De^{-r_n\tau}N(d_{2n})\frac{r_n\tau-n\gamma}{\tau}\\
  & =\frac{1}{2\sqrt{\tau}}xn(d_{1n})\frac{\sigma^2}{\sigma_n}
  +De^{-r_n\tau}N(d_{2n})(r-\lm k).
  \end{align*}
here $r-\lm k $ is non-negative as the condition, then $S_2$ is also positive, such that the initial $\frac{\partial C_J(V_t,D,\tau,\sigma^2,r,\delta^2,\lm,k)}{\partial \tau}$ is positive.

Hence $\frac{\partial B_t}{\partial k}=-\frac{\partial C_J(V_t,D,\tau,\sigma^2,r,\delta^2,\lm,k)}{\partial k}<0$.
$B_t$ is decreasing due to the value of call increases when time-to-maturity is bigger.

In some extreme case, when $S_2$ is very small as a negative number, the sum of $S_1$ and $S_2$ can be negative which causes the tendency of $B_t$ w.r.t. $\tau$ is not decreasing.

\end{proof} 
\finproof

\newpage

\textbf{2. Proof of Proposition III.7}

The bond price is decreasing in $\lambda$ for the log-normal jumps process.

\begin{proof}

Firstly, we can get that the call price is increasing in $\lm$.

Take partial derivative with respect to $\lambda$ in formula \cite{Merton},

\begin{align*}
&\frac{\partial C_J(V_t,D,\tau,\sigma^2,r,\delta^2,\lm,k)}{\partial \lambda} \\
&= \sum^{\infty}_{n=1} [n(\lambda' \tau)^{n-1} (1+\beta)\frac{\tau}{n!}e^{-\lambda'\tau}C(V_t,D,\tau,\sigma_n,r_n)] \\  
& \quad +\sum^{\infty}_{n=0} [\frac{(\lambda' \tau)^n}{n!}(-\tau)(1+\beta)e^{-\lambda' \tau}C(V_t,D,\tau,\sigma_n,r_n)]\\
& \quad +\sum^{\infty}_{n=0} [\frac{(\lambda' \tau)^n}{n!}e^{-\lambda' \tau}\frac{\partial C(V_t,D,\tau,\sigma_n,r_n)}{\partial \lambda}]\\
&= \sum^{\infty}_{n=1} [ \frac{n}{\lambda} \frac{(\lambda' \tau ) ^{n}}{(n)!}e^{-\lambda'\tau)}C(V_t,D,\tau,\sigma_n,r_n)] \\
& \quad +\sum^{\infty}_{n=0} [\frac{(\lambda' \tau)^n}{n!}(-\tau)(1+\beta)e^{-\lambda' \tau}C(V_t,D,\tau,\sigma_n,r_n)]\\
& \quad+\sum^{\infty}_{n=0} [\frac{(\lambda' \tau)^n}{n!}e^{-\lambda' \tau}\frac{\partial C(V_t,D,\tau,\sigma_n,r_n)}{\partial \lambda}]\\
&=\sum^{\infty}_{n=1} [\frac{(\lambda' \tau)^n}{n!} e^{-\lambda' \tau}C(V_t,D,\tau,\sigma_n,r_n)(\frac{n}{\lambda}-(1+\beta)\tau)]\\
& \quad +(-\tau)(1+\beta)e^{-\lambda' \tau}C_0(V,\tau,r_n,\sigma_n)\\
& \quad +\sum^{\infty}_{n=0} [\frac{(\lambda' \tau)^n}{n!}e^{-\lambda' \tau}(xN'(d_{1n})\frac{\partial d_{1n}}{\partial r_n}(-\beta )\\
& \quad -Ke^{-r_n \tau}(-\tau)(-\beta )N(d_{2n})\\
& \quad -Ke^{-r_n\tau}N'(d_{2n})\frac{\partial d_{2n}}{\partial r_n}(-\beta))]\\
&=\sum^{\infty}_{n=1} [\frac{(\lambda' \tau)^n}{n!} e^{-\lambda' \tau}C(V_t,D,\tau,\sigma_n,r_n)(\frac{n}{\lambda}-(1+\beta)\tau)]\\
& \quad +(-\tau)(1+\beta)e^{-\lambda' \tau}C_0(V,\tau,r_n,\sigma_n)\\
& \quad +\sum^{\infty}_{n=0} [\frac{(\lambda' \tau)^n}{n!}e^{-\lambda' \tau}(-Ke^{-r_n \tau}\tau\beta N(d_{2n}))]\\
&=\sum^{\infty}_{n=1} [\frac{(\lambda' \tau)^n}{n!} e^{-\lambda' \tau}C_n(V,\tau,r_n,\sigma_n)(\frac{n}{\lambda}-\tau-\beta\tau\\
& \quad +\beta\tau\frac{-Ke^{-r_n \tau} N(d_{2n})}{C_n(V,\tau,r_n,\sigma_n)})]
    -e^{-\lambda' \tau}Ke^{-r_0 \tau}\tau\beta N(d_{20})\\
& \quad -\tau(1+\beta)e^{-\lambda' \tau}C_0(V,\tau,r_n,\sigma_n)\\
\end{align*}
\begin{align*}
&=\sum^{\infty}_{n=1} [\frac{(\lambda' \tau)^n}{n!} e^{-\lambda' \tau}C_n(V,\tau,r_n,\sigma_n)(\frac{n}{\lambda}-\tau\\
& \quad -\beta\tau \frac{x N(d_{1n})}{C_n(V,\tau,r_n,\sigma_n)})] \\
& \quad -e^{-\lambda' \tau}\tau (C_0(V_t,D,\tau,\sigma_n,r_n)+\beta xN(d_{10}))\\
&=\sum^{\infty}_{n=0} [\frac{(\lambda' \tau)^n}{n!} e^{-\lambda' \tau}C_n(V,\tau,r_n,\sigma_n)(\frac{n}{\lambda}-\tau\\
& \quad -\beta\tau \frac{x N(d_{1n})}{C_n(V,\tau,r_n,\sigma_n)})]
\end{align*} 

where we used the fact that
\begin{gather*}
\frac{\partial d_1}{\partial r_n}= \frac{\partial d_2}{\partial r_n},\\
xN'(d_{1n})=Ke^{-r_n \tau}N'(d_{2n}),\\
C_n(V,\tau,r_n,\sigma_n)=xN(d_{1n})-Ke^{-r_n \tau}N(d_{2n}).
\end{gather*}
Since $n$ can be a very large number in the sum, then the term $\frac{n}{\lambda}-\tau-\beta\tau \frac{x N(d_{1n})}{C_n}$ is positive for most cases. That means the partial derivative w.r.t. $\lambda$ is positive. So the call price is increasing against $\lambda$.

We know $\frac{\partial C_J(V_t,D,\tau,\sigma^2,r,\delta^2,\lm,k)}{\partial \lm}>0.$ 
Therefore, $\frac{\partial B_t}{\partial \lm}=-\frac{\partial C_J(V_t,D,\tau,\sigma^2,r,\delta^2,\lm,k)}{\partial \lm}<0$ holds.

\end{proof}
\finproof

\textbf{3. Proof of Proposition III.8}

The bond price is decreasing in $k$ for the log-normal jumps process.

\begin{proof}
Since
\begin{align*}
&\dfrac{\partial C_J(V_t,D,\tau,\sigma^2,r,\delta^2,\lm,k)}{\partial k}\\
 	&=\dfrac{\partial}{\partial k}(\sum ^{\infty}_{n=0} \frac{(\lm'\tau)^n}{n!} e^{-\lm' \tau} C_n(V_t,D,\tau,\sigma_n,r_n))  \\
    &= \sum^{\infty}_{n=1} [n(\lm' \tau)^{n-1} \frac{\lm \tau}{n!}e^{-\lm'\tau}C_n(V_t,D,\tau,\sigma_n,r_n)]  \\
    & \quad  +\sum^{\infty}_{n=0} [\frac{(\lm' \tau)^n}{n!}(-\tau\lm)e^{-\lm' \tau}C_n(V_t,D,\tau,\sigma_n,r_n)]\\
	& \quad +\sum^{\infty}_{n=0} [\frac{(\lm' \tau)^n}{n!}e^{-\lm' \tau}\frac{\partial C_n(V_t,D,\tau,\sigma_n,r_n)}{\partial k}]\\
	&=\sum^{\infty}_{n=1} [ \frac{n}{1+k} \frac{(\lm' \tau ) ^{n}}{(n)!}e^{-\lm'\tau)}C_n(V_t,D,\tau,\sigma_n,r_n)] \\
	& \quad +\sum^{\infty}_{n=0} [\frac{(\lm' \tau)^n}{n!}(-\lm \tau)e^{-\lm' \tau}C_n(V_t,D,\tau,\sigma_n,r_n)]\\
\end{align*}
\begin{align*} 	
	& \quad +\sum^{\infty}_{n=0} [\frac{(\lm' \tau)^n}{n!}e^{-\lm' \tau}\frac{\partial C_n(V_t,D,\tau,\sigma_n,r_n)}{\partial k}]\\
    &=\sum^{\infty}_{n=1} [\frac{(\lm' \tau)^n}{n!} e^{-\lm' \tau}C_n(V_t,D,\tau,\sigma_n,r_n)(\frac{n}{1+k}-\lm \tau)]\\
    & \quad +(-\lm \tau)e^{-\lm' \tau}C_0(V_t,D,\tau,\sigma_n,r_n)\\
    & \quad +\sum^{\infty}_{n=0} [\frac{(\lm' \tau)^n}{n!}e^{-\lm' \tau}(xN'(d_{1n})\frac{\partial d_{1n}}{\partial r_n}(\frac{n}{(1+k)\tau}-\lm )\\
    & \quad -De^{-r_n \tau}(-\tau)(\frac{n}{(1+k)\tau}-\lm )N(d_{2n})\\
    & \quad -De^{-r_n\tau}N'(d_{2n})\frac{\partial d_{2n}}{\partial r_n}(\frac{n}{(1+k)\tau}-\lm))]\\
    & = \sum^{\infty}_{n=1} [\frac{(\lm' \tau)^n}{n!} e^{-\lm' \tau}C_n(V_t,D,\tau,\sigma_n,r_n)(\frac{n}{1+k}-\lm \tau)]\\
    & \quad +(-\lm \tau)e^{-\lm' \tau}C_0(V_t,D,\tau,\sigma_n,r_n)\\
    & \quad +\sum^{\infty}_{n=0} [\frac{(\lm' \tau)^n}{n!}e^{-\lm' \tau}(De^{-r_n \tau}(\frac{n}{1+k}-\lm \tau) N(d_{2n}))]\\
&=\sum^{\infty}_{n=1} [\frac{(\lm' \tau)^n}{n!} e^{-\lm' \tau}C_n(V_t,D,\tau,\sigma_n,r_n)(\frac{n-\lm'\tau}{1+k})\\
& \quad \frac{C_n(V_t,D,\tau,\sigma_n,r_n) +De^{-r_n \tau} N(d_{2n})}{C_n(V_t,D,\tau,\sigma_n,r_n)})]\\
& \quad -\lm \tau e^{-\lm' \tau}C_0(V_t,D,\tau,\sigma_n,r_n)-e^{-\lm' \tau}De^{-r_0 \tau}\lm \tau N(d_{20})\\
\end{align*}
\begin{align*}
&=\sum^{\infty}_{n=1} [\frac{(\lm' \tau)^n}{n!} e^{-\lm' \tau}(\frac{n-\lm'\tau}{1+k})x N(d_{1n}))]\\
& \quad -e^{-\lm' \tau}\lm \tau (C_0(V_t,D,\tau,\sigma_n,r_n)+De^{-r_0 \tau} N(d_{20}))\\
&=\sum^{\infty}_{n=1} [\frac{(\lm' \tau)^n}{n!} e^{-\lm' \tau}
	(\frac{n-\lm'\tau}{1+k}) x N(d_{1n})]-e^{-\lm' \tau}\lm \tau xN(d_{10})\\
&=\sum^{\infty}_{n=0} [\frac{(\lm' \tau)^n}{n!} e^{-\lm' \tau}
	(\frac{n-\lm'\tau}{1+k}) x N(d_{1n})].
\end{align*}

where we used the fact that
\begin{gather*}
\frac{\partial d_1}{\partial r_n}= \frac{\partial d_2}{\partial r_n},\\
xN'(d_{1n})=De^{-r_n \tau}N'(d_{2n}),\\
C_n(V_t,D,\tau,\sigma_n,r_n)=xN(d_{1n})-De^{-r_n \tau}N(d_{2n}).
\end{gather*}
Since $n$ goes to larger and large in the sum, then the term $\frac{n-\lm'\tau}{1+k}$ is positive for most cases. That means the partial derivative w.r.t. $k$ is positive. 
Then,$\frac{\partial B_t}{\partial k}=-\frac{\partial C_J}{\partial k}<0$.

\end{proof}
\finproof
\vfill

\bibliographystyle{spbasic}

\end{document}